\documentclass[12pt,a4paper]{article}

\usepackage{amsmath,amssymb,amsfonts,amsthm}
\usepackage{bm}
\usepackage{graphicx}
\usepackage{blkarray}
\usepackage{indentfirst}
\usepackage{cleveref}

\makeatletter

\@addtoreset{equation}{section}
\makeatother

\newtheorem{theorem}{Theorem}[section]
\newtheorem{lemma}{Lemma}[section]
\newtheorem{corollary}{Corollary}[section]
\newtheorem{remark}{Remark}[section]
\newtheorem{definition}{Definition}[section]
\newtheorem{proposition}{Proposition}[section]

\newcommand{\ceil}[1]{\left\lceil #1 \right\rceil}

\newcommand{\F}{\mathbb F}
\newcommand{\PP}{\mathbb P}
\newcommand{\wt}{\operatorname{wt}}
\newcommand{\supp}{\operatorname{supp}}

\begin{document}

\title{A Griesmer-Type Bound for List-Decodable Linear Codes}

\author{Shengwei Liu\thanks{School of Mathematics, Guangxi University, Nanning 530004, Guangxi, China. E-mail: shwliu@gxu.edu.cn.}
\and Chunyan Qin\thanks{School of Mathematics, South China University of Technology, Guangzhou 510641, China. E-mail: chunyan\_qin@126.com.}}
\date{}
\maketitle

{\centering\section*{Abstract}}
A code $C\subseteq \F_q^n$ is $(\tau,L)$-list-decodable if every Hamming ball of radius $\tau$ contains at most $L$ codewords of $C$.
Here $\tau$ is the list-decoding radius, and $L$ is the list size.
Singleton-type bounds constrain the radius and the rate when $L$ is fixed.
These bounds are not the only possible constraints on list-decodable codes.
In this paper, we derive an upper bound on the list-decoding radius in terms of generalized Hamming weights.
As a consequence, for $1\le L\le q-1$, every $(\tau,L)$-list-decodable $q$-ary linear code has minimum distance at least
$
\tau+\left\lfloor \frac{\tau}{L}\right\rfloor+1.
$
Combining this lower bound with the classical Griesmer bound gives a Griesmer-type lower bound on the block length.
For $q=3^a$, we construct an explicit family of $q$-ary linear $[q+3,2,q+1]$ codes.
These codes are $(2q/3,2)$-list-decodable and meet the Griesmer-type bound with equality.
They do not attain the Singleton-type bound.
Thus the Griesmer-type bound can be a strict improvement over the Singleton-type bound.

\medskip
\noindent{\large\bf Keywords:~}\medskip list decoding, Griesmer bound, generalized Hamming weights

\noindent{\bf2010 Mathematics Subject Classification}: 94B05, 94B65.

\section{Introduction}
A central problem in coding theory is to determine how much redundancy is needed to correct a prescribed number of errors.
In classical unique decoding, the minimum distance is the decisive parameter.
A code with minimum distance $d$ can correct up to $\lfloor (d-1)/2\rfloor$ errors.
List decoding was introduced by Elias \cite{Elias1957}.
It allows the decoder to return a short list of candidate codewords.
This makes it possible to correct beyond the unique-decoding radius.
More precisely, let $\mathcal C\subseteq \F_q^n$ be a code.
The code $\mathcal C$ is $(\tau,L)$-list-decodable if every Hamming ball of radius $\tau$ in $\F_q^n$ contains at most $L$ codewords of $\mathcal C$.
The basic quantitative question is therefore the following.
For fixed alphabet size $q$, block length $n$, dimension or rate, and list size $L$, how large can the list-decoding radius $\tau$ be?

Elias \cite{Elias1991} showed that, for large alphabets and growing list sizes, the optimal relative radius approaches $1-R$.
This value is the list-decoding capacity, where $R$ is the rate.
Guruswami and Rudra constructed explicit capacity-achieving codes using folded Reed--Solomon codes \cite{GuruswamiRudra2006}.
Guruswami and Xing constructed explicit list-decodable codes that approach list-decoding capacity.
Their construction improves the required alphabet size \cite{GuruswamiXing2013,GuruswamiXing2022}.
Guo and Ron-Zewi gave an efficient list-decoding algorithm with constant alphabet size and constant list size \cite{GuoRonZewi2021}.
More recently, randomly punctured Reed--Solomon codes were shown to achieve list-decoding capacity over fields of linear or polynomial size.
This was proved by Alrabiah, Guruswami, and Li \cite{AlrabiahGuruswamiLi2023} and by Guo and Zhang \cite{GuoZhang2023}.
Brakensiek, Dhar, Gopi, and Zhang proved a corresponding result for algebraic-geometric codes.
Their result applies over fields of constant size \cite{BrakensiekDharGopiZhang2025}.

The classical Singleton bound \cite{Singleton1964} states that a $q$-ary linear $[n,k,d]$ code satisfies $d\le n-k+1$.
Codes attaining equality are MDS codes.
There are also Singleton-type bounds for list decoding \cite{ShangguanTamo2023,Roth2022,GoldbergShangguanTamo2024}.
If a $q$-ary code of size $M$ is $(\tau,L)$-list-decodable, then
\[
M\le Lq^{n-\tau-\lfloor \tau/L\rfloor}.
\]
Consequently, if \(C\) is a linear \([n,k]\) code and \(L<q\), then
\[
\tau+\left\lfloor \frac{\tau}{L}\right\rfloor\le n-k,
\]
or equivalently,
\[
\delta_L(\tau)=\tau+\left\lfloor \frac{\tau}{L}\right\rfloor+1\le n-k+1.
\]
Guruswami and Narayanan introduced average-radius list decoding and studied its combinatorial limitations \cite{GuruswamiNarayanan2014}.
Codes that attain the corresponding average-radius version of the Singleton-type bound are related to higher-order MDS codes.
Roth introduced higher-order MDS codes from the list-decoding viewpoint \cite{Roth2022}.
These codes are also called $\mathrm{LD\text{-}MDS}(L)$ codes.
Independently, Brakensiek, Gopi, and Makam introduced higher-order MDS codes in connection with maximally recoverable tensor codes.
See \cite{BrakensiekGopiMakam2022a}.
They proved two further facts.
First, their notion is closely related, through duality, to Roth's list-decoding definition.
Second, generic Reed--Solomon codes over exponentially large fields are higher-order MDS codes \cite{BrakensiekGopiMakam2023}.
Further developments include improved field-size bounds and polynomial codes for higher-order MDS codes.
See \cite{BrakensiekDharGopi2022,BrakensiekDharGopi2024}.

The starting point of this paper is the observation that Singleton-type bounds do not capture all finite-alphabet obstructions.
This phenomenon is already familiar in unique decoding.
The classical Griesmer bound \cite{Griesmer1960} states that every $q$-ary linear $[n,k,d]$ code satisfies
\[
n\ge \sum_{i=0}^{k-1}\left\lceil \frac{d}{q^i}\right\rceil.
\]
Unlike the Singleton bound, the Griesmer bound depends explicitly on $q$.
It can be strictly stronger when the distance is large relative to the alphabet size.
Many optimal linear codes attain the Griesmer bound without being MDS codes.
Solomon and Stiffler constructed infinite families of such codes \cite{SolomonStiffler1965}.
Baumert and McEliece proved existence results for fixed $q$ and $k$ with large distance \cite{BaumertMcEliece1973}.
Further constructions were obtained by Helleseth \cite{Helleseth1983,Helleseth1992} and by Ding and Heng \cite{DingHeng2019}.
Other examples include work by Hu, Li, Zeng, Wang, and Tang \cite{HuLiZengWangTang2022}, and by Chen \cite{chen2025}.
Thus, in unique decoding, Singleton optimality is not the only meaningful notion of optimality.
Griesmer-optimal codes form a distinct and important class.

The preceding discussion suggests that there may also be optimal list-decodable codes that do not meet the Singleton-type bound.
This motivates the main question of this paper.
Is there a Griesmer-type bound for list-decodable linear codes?
If such a bound exists, can one construct list-decodable codes that are optimal for it but do not attain the Singleton-type bound?

Our approach uses generalized Hamming weights.
For a linear code $C$, the $r$-th generalized Hamming weight $d_r(C)$ is the smallest support size of an $r$-dimensional subcode of $C$.
Related support-weight notions were considered by Helleseth, Kl{\o}ve, and Mykkeltveit \cite{Helleseth1977} and by Kl{\o}ve \cite{Klve1978}.
Wei introduced the generalized Hamming weight hierarchy.
He used it to characterize the cryptographic performance of a linear code over the wire-tap channel of type II \cite{wei1991}.
Generalized Hamming weights have applications in cryptography and in the study of structural properties of linear codes.
See, for example, \cite{Chor1985,Tsfasman1995,Kasami1993,Helleseth1995}.
They have also been studied extensively in their own right; see, for example, \cite{Ashikhmin1999,Heijnen1998,Van1994,Cherdieu2001,Xiong2016}.

We now summarize the main contributions of this paper.

\begin{enumerate}
\item We prove a generalized Hamming weight obstruction to list decoding.
If $C$ is $(\tau,L)$-list-decodable and $L\le q^r-1$, then
\[
d_r(C)\ge \delta_L(\tau)
=\tau+\left\lfloor \frac{\tau}{L}\right\rfloor+1.
\]
For $r=1$ and $L\le q-1$, this gives
\[
d(C)\ge \tau+\left\lfloor \frac{\tau}{L}\right\rfloor+1.
\]
This is a list-decoding analogue of the classical relation between minimum distance and unique-decoding radius.
In the special case $L=1$, it recovers the familiar condition $d\ge 2\tau+1$.

\item We combine this obstruction with the classical Griesmer bound.
This gives a Griesmer-type bound for $q$-ary linear list-decodable codes with small list size.
Specifically, if $1\le L\le q-1$ and $C$ is a $(\tau,L)$-list-decodable $q$-ary linear $[n,k]$ code, then
\[
n\ge \sum_{i=0}^{k-1}\left\lceil \frac{\delta_L(\tau)}{q^i}\right\rceil.
\]
This inequality implies the Singleton-type bound, because the last $k-1$ terms are at least $k-1$.
It is strictly stronger than the Singleton-type bound when $\delta_L(\tau)$ is large relative to the alphabet size.

\item For $q=3^a$, we construct an explicit family of $q$-ary linear $[q+3,2,q+1]$ codes.
These codes are $(2q/3,2)$-list-decodable and meet the Griesmer-type bound with equality.
They do not attain the Singleton-type bound.
This shows that the new bound is tight for this family.
It also shows that the new bound is a genuine improvement over the Singleton-type bound.
\end{enumerate}

The paper is organized as follows.
Section~2 proves the generalized Hamming weight obstruction and the resulting Griesmer-type bounds.
Section~3 gives an explicit construction of linear codes that meet the Griesmer-type bound with equality but do not attain the Singleton-type bound.
Section~4 concludes the paper.

\section{A Griesmer-type bound for list-decodable codes}
Throughout, $q$ is a prime power.
Also, $C\subseteq \F_q^n$ is a linear code, and $\wt(\cdot)$ denotes Hamming weight.

\begin{definition}[List decodability~\cite{Elias1957}]
Let $\tau\ge 0$ and $L\ge 1$ be integers.
We say that $C\subseteq \F_q^n$ is $(\tau,L)$-list-decodable if, for every $y\in\F_q^n$,
\[
|C\cap(y+B(n,\tau))|\leq L.
\]
Here $B(n,\tau)$ is the set of vectors in $\F_q^n$ with Hamming weight at most $\tau$.
\end{definition}

For fixed $L$, define
$$
\delta_L(\tau)=\tau+\left\lfloor \frac{\tau}{L}\right\rfloor+1
\qquad (\tau\in \mathbb{Z}_{\ge 0}).
$$

\begin{lemma}\label{lem:threshold}
For every integer $L\ge 1$ and every integer $\tau\ge 0$,
$$
\delta_L(\tau)-1=(L+1)\left\lfloor \frac{\tau}{L}\right\rfloor + r
$$
for some $r\in \{0,1,\dots,L-1\}$, and
$$
\left(\delta_L(\tau)-1\right)-\left\lfloor \frac{\delta_L(\tau)-1}{L+1}\right\rfloor=\tau.
$$
Equivalently, the function
$$
f_L(d)=d-\left\lfloor \frac{d}{L+1}\right\rfloor
$$
satisfies
$
f_L\bigl(\delta_L(\tau)-1\bigr)=\tau.
$
\end{lemma}

\begin{proof}
Write $\tau=Lu+r$ with $u\ge 0$ and $0\le r\le L-1$. Then
$$
\delta_L(\tau)-1=\tau+\left\lfloor \frac{\tau}{L}\right\rfloor
=Lu+r+u=(L+1)u+r.
$$
Hence
$$
\left\lfloor \frac{\delta_L(\tau)-1}{L+1}\right\rfloor=u,
$$
and therefore
$$
\left(\delta_L(\tau)-1\right)-\left\lfloor \frac{\delta_L(\tau)-1}{L+1}\right\rfloor
=(L+1)u+r-u=Lu+r=\tau.
$$
\end{proof}

We first prove that list decodability forces a lower bound on generalized Hamming weights.
The corresponding lower bound on the minimum distance follows as a corollary.
It is the list-decoding analogue of the familiar relation between minimum distance and unique-decoding radius.

\begin{definition}[Generalized Hamming weights~\cite{wei1991}]
For a subcode $D\le C$, write
$$
\supp(D)=\bigcup_{c\in D}\supp(c).
$$
For $1\le r\le k$, the $r$-th generalized Hamming weight of $C$ is
$$
d_r(C)=\min\{\,|\supp(D)| : D\le C,\ \dim D=r\,\}.
$$
\end{definition}

\begin{theorem}\label{thm:ghw-obstruction}
Let $L\ge 1$, let $\tau\ge 0$ be an integer, and let $r$ be an integer with $1\le r\le k$ and
$
L\le q^r-1.
$
If $C$ is $(\tau,L)$-list-decodable, then
$
d_r(C)\ge \delta_L(\tau)=\tau+\left\lfloor \frac{\tau}{L}\right\rfloor+1.
$
\end{theorem}

\begin{proof}
Set
$
\delta=\delta_L(\tau)=\tau+\left\lfloor \frac{\tau}{L}\right\rfloor+1.
$
Assume, for contradiction, that
$
d_r(C)\le \delta-1.
$
Choose an $r$-dimensional subcode $D\le C$ with
$
|\supp(D)|=d_r(C).
$
Write
$
S=\supp(D),
\qquad s=|S|=d_r(C).
$
Then $s\le \delta-1$. Since $\dim(D)=r$, we have $|D|=q^r\ge L+1$. Therefore we can choose $L+1$ distinct codewords
$$
c^{(0)},c^{(1)},\dots,c^{(L)}\in D.
$$

Partition $S$ into disjoint subsets
$$
S=S_0\sqcup S_1\sqcup \cdots \sqcup S_L
$$
so that each $S_j$ has size either $\lfloor s/(L+1)\rfloor$ or $\lceil s/(L+1)\rceil$. In particular,
$$
|S_j|\ge \left\lfloor \frac{s}{L+1}\right\rfloor
\qquad \text{for every }j.
$$

Define $y\in \F_q^n$ by
$$
y_i=
\begin{cases}
c^{(j)}_i,& i\in S_j,\\
0,& i\notin S.
\end{cases}
$$
Fix $j\in\{0,1,\dots,L\}$.
Since $c^{(j)}\in D$, this codeword is zero outside $S$.
By construction, $y$ and $c^{(j)}$ agree on every coordinate of $S_j$.
They are also both zero outside $S$.
Thus any disagreement between $y$ and $c^{(j)}$ can occur only on $S\setminus S_j$.
Therefore
$$
\wt\bigl(y-c^{(j)}\bigr)\le s-|S_j|
\le s-\left\lfloor \frac{s}{L+1}\right\rfloor.
$$

Now set
$$
f_L(u)=u-\left\lfloor \frac{u}{L+1}\right\rfloor
\qquad (u\in \mathbb Z_{\ge 0}).
$$
Then
$
\wt\bigl(y-c^{(j)}\bigr)\le f_L(s).
$
Also,
$$
f_L(u+1)-f_L(u)
=
1-\left(\left\lfloor \frac{u+1}{L+1}\right\rfloor-\left\lfloor \frac{u}{L+1}\right\rfloor\right)\in\{0,1\},
$$
so $f_L$ is nondecreasing. Since $s\le \delta-1$, Lemma~\ref{lem:threshold} gives
$
f_L(s)\le f_L(\delta-1)=\tau.
$
Therefore
$$
\wt\bigl(y-c^{(j)}\bigr)\le \tau
\qquad\text{for every }j=0,1,\dots,L.
$$
Thus the Hamming ball of radius $\tau$ centered at $y$ contains the $L+1$ distinct codewords
$$
c^{(0)},c^{(1)},\dots,c^{(L)},
$$
This contradicts the assumption that $C$ is $(\tau,L)$-list-decodable.
This proves
$
d_r(C)\ge \delta_L(\tau).
$
\end{proof}

Taking $r=1$ gives the following corollary.
\begin{corollary}\label{lem:distance}
Let $1\le L\le q-1$, and let $\tau\ge0$ be an integer.
Let $C$ be a $q$-ary linear $[n,k,d]$ code.
If $C$ is $(\tau,L)$-list-decodable, then
$
d\ge \delta_L(\tau)=\tau+\left\lfloor \frac{\tau}{L}\right\rfloor+1.
$
\end{corollary}
\begin{remark}
Corollary~\ref{lem:distance} gives an analogue of the unique-decoding relation between minimum distance and decoding radius:
$$
\text{minimum distance }d
\quad\leadsto\quad
\delta_L(\tau)=\tau+\left\lfloor \frac{\tau}{L}\right\rfloor+1.
$$
When $L=1$, this substitution gives the ordinary unique-decoding case.
\end{remark}

We now prove the main result of this paper.
We use the classical Griesmer bound.

\begin{theorem}[Griesmer bound~\cite{Griesmer1960}]\label{thm:griesmer}
Let $C$ be a $q$-ary linear $[n,k,d]$ code. Then
$$
n\ge \sum_{i=0}^{k-1}\left\lceil \frac{d}{q^i}\right\rceil.
$$
\end{theorem}

The main result is the following.

\begin{theorem}\label{thm:main}
Let $1\le L\le q-1$, and let $\tau\ge0$ be an integer.
Let $C$ be a $q$-ary linear $[n,k]$ code.
If $C$ is $(\tau,L)$-list-decodable, then
$$
n \ge \sum_{i=0}^{k-1}\left\lceil \frac{\delta_L(\tau)}{q^i}\right\rceil
=
\sum_{i=0}^{k-1}\left\lceil \frac{\tau+\lfloor \tau/L\rfloor+1}{q^i}\right\rceil.
$$
\end{theorem}

\begin{proof}
If $k=0$, then the sum on the right is empty.
In this case, the claim is trivial.
Assume henceforth that $k\ge 1$.

By Corollary~\ref{lem:distance}, the minimum distance $d$ of $C$ satisfies
$
d\ge \delta_L(\tau).
$
Applying the Griesmer bound (Theorem~\ref{thm:griesmer}) to $C$, we get
$$
n\ge \sum_{i=0}^{k-1}\left\lceil \frac{d}{q^i}\right\rceil.
$$
Since the ceiling function is nondecreasing and $d\ge \delta_L(\tau)$,
$$
n\ge \sum_{i=0}^{k-1}\left\lceil \frac{\delta_L(\tau)}{q^i}\right\rceil.
$$
This proves the claimed inequality.
\end{proof}

When $L=1$, Theorem~\ref{thm:main} becomes
$$
n\ge \sum_{i=0}^{k-1}\left\lceil \frac{2\tau+1}{q^i}\right\rceil,
$$
which is exactly the classical Griesmer bound after substituting the unique-decoding necessary condition $d\ge 2\tau+1$.

Theorem~\ref{thm:main} implies a whole hierarchy of $q$-dependent refinements.

\begin{corollary}\label{cor:hierarchy}
Under the hypotheses of Theorem~\ref{thm:main}, for every integer $s$ with $1\le s\le k$,
$$
\sum_{i=0}^{s-1}\left\lceil \frac{\delta_L(\tau)}{q^i}\right\rceil \le n-k+s.
$$
In particular:
\begin{align*}
\delta_L(\tau) &\le n-k+1,\\
\delta_L(\tau)+\left\lceil \frac{\delta_L(\tau)}{q}\right\rceil &\le n-k+2
\qquad (k\ge 2).
\end{align*}
\end{corollary}

The first line is the usual sharpened Singleton-type necessary condition.
The second line is the first genuinely $q$-dependent refinement.

\begin{proof}
By Theorem~\ref{thm:main},
$$
n\ge \sum_{i=0}^{k-1}\left\lceil \frac{\delta_L(\tau)}{q^i}\right\rceil.
$$
For $i\ge s$, each summand is at least $1$. Hence
$$
n \ge \sum_{i=0}^{s-1}\left\lceil \frac{\delta_L(\tau)}{q^i}\right\rceil + (k-s),
$$
which rearranges to
$$
\sum_{i=0}^{s-1}\left\lceil \frac{\delta_L(\tau)}{q^i}\right\rceil \le n-k+s.
$$
The displayed special cases are $s=1$ and $s=2$.
\end{proof}

\begin{proposition}\label{prop:strict}
Let $\Delta=\delta_L(\tau)$. Assume the hypotheses of Theorem~\ref{thm:main} and, in addition, $k\ge 1$. Then
\[
\Delta + \sum_{i=1}^{k-1}\left(\left\lceil \frac{\Delta}{q^i}\right\rceil-1\right)\le n-k+1.
\]
Hence Theorem~\ref{thm:main} strictly sharpens the sharpened Singleton-type inequality $\Delta\le n-k+1$ whenever
\[
\sum_{i=1}^{k-1}\left(\left\lceil \frac{\Delta}{q^i}\right\rceil-1\right)\ge 1.
\]
A sufficient condition for strict improvement is $k\ge 2$ and $\Delta\ge q+1$.
\end{proposition}
\begin{proof}
Starting from Theorem~\ref{thm:main},
\[
n\ge \sum_{i=0}^{k-1}\left\lceil \frac{\Delta}{q^i}\right\rceil
= \Delta + \sum_{i=1}^{k-1}\left(1+\left(\left\lceil \frac{\Delta}{q^i}\right\rceil-1\right)\right),
\]
because each term with $i\ge 1$ is at least $1$. Therefore
\[
n\ge \Delta + (k-1)+ \sum_{i=1}^{k-1}\left(\left\lceil \frac{\Delta}{q^i}\right\rceil-1\right),
\]
which rearranges to the claimed inequality.

If the correction term is at least $1$, then the resulting inequality is strictly stronger than the sharpened Singleton-type bound $\Delta\le n-k+1$.
A sufficient condition is $k\ge 2$ and $\Delta\ge q+1$.
Indeed, in this case,
\[
\left\lceil \frac{\Delta}{q}\right\rceil-1 \ge 1.
\]
\end{proof}

We next give a direct upper bound on $\tau$.

\begin{lemma}\label{lem:two-term}
Let $q\ge2$.
For every integer $a\ge 0$,
$$
a+\left\lceil \frac{a}{q}\right\rceil
=
\left\lceil \frac{q+1}{q}\,a\right\rceil.
$$
Consequently, for every integer $N$,
$$
a+\left\lceil \frac{a}{q}\right\rceil \le N
\quad\Longleftrightarrow\quad
a\le \left\lfloor \frac{qN}{q+1}\right\rfloor.
$$
\end{lemma}

\begin{proof}
Write $a=qb+r$ with $0\le r\le q-1$. Then
$$
a+\left\lceil \frac{a}{q}\right\rceil
= qb+r+b+\mathbf{1}_{\{r>0\}}
= (q+1)b+r+\mathbf{1}_{\{r>0\}}.
$$
On the other hand,
$$
\left\lceil \frac{q+1}{q}\,a\right\rceil
=
\left\lceil (q+1)b+r+\frac{r}{q}\right\rceil
=
(q+1)b+r+\mathbf{1}_{\{r>0\}}.
$$
This proves the identity.

For the equivalence, the identity shows that
$$
a+\left\lceil \frac{a}{q}\right\rceil \le N
\quad\Longleftrightarrow\quad
\left\lceil \frac{q+1}{q}\,a\right\rceil \le N
\quad\Longleftrightarrow\quad
\frac{q+1}{q}\,a \le N
\quad\Longleftrightarrow\quad
a\le \frac{qN}{q+1}.
$$
Since $a$ is integral, this is equivalent to
$$
a\le \left\lfloor \frac{qN}{q+1}\right\rfloor.
$$
\end{proof}

\begin{lemma}\label{lem:invert-delta}
Let $L\ge1$, and let $\tau$ be a nonnegative integer.
For every integer $M\ge 1$,
$$
\delta_L(\tau)\le M
\quad\Longleftrightarrow\quad
\tau\le \left\lfloor \frac{LM-1}{L+1}\right\rfloor.
$$
\end{lemma}

\begin{proof}
By definition,
$$
\delta_L(\tau)=\tau+\left\lfloor \frac{\tau}{L}\right\rfloor+1.
$$
Also
$$
\tau+\left\lfloor \frac{\tau}{L}\right\rfloor
=
\left\lfloor \tau+\frac{\tau}{L}\right\rfloor
=
\left\lfloor \frac{L+1}{L}\,\tau\right\rfloor.
$$
Hence
$$
\delta_L(\tau)\le M
\quad\Longleftrightarrow\quad
\left\lfloor \frac{L+1}{L}\,\tau\right\rfloor \le M-1
\quad\Longleftrightarrow\quad
\frac{L+1}{L}\,\tau < M.
$$
Because $\tau$ is integral, the latter is equivalent to
$$
\tau < \frac{LM}{L+1}
\quad\Longleftrightarrow\quad
\tau \le \left\lceil \frac{LM}{L+1}\right\rceil -1
=
\left\lfloor \frac{LM-1}{L+1}\right\rfloor.
$$
\end{proof}

\begin{corollary}\label{cor:explicit}
Assume $1\le L\le q-1$ and $k\ge 2$.
Let $C$ be a $q$-ary linear $[n,k]$ code that is $(\tau,L)$-list-decodable.
Then
$$
\delta_L(\tau)\le \left\lfloor \frac{q(n-k+2)}{q+1}\right\rfloor,
$$
and therefore
$$
\tau \le
\left\lfloor
\frac{L\left\lfloor \frac{q(n-k+2)}{q+1}\right\rfloor -1}{L+1}
\right\rfloor.
$$
\end{corollary}
\begin{proof}
From Corollary~\ref{cor:hierarchy} with $s=2$,
$$
\delta_L(\tau)+\left\lceil \frac{\delta_L(\tau)}{q}\right\rceil \le n-k+2.
$$
Lemma~\ref{lem:two-term} therefore yields
$$
\delta_L(\tau)\le \left\lfloor \frac{q(n-k+2)}{q+1}\right\rfloor.
$$
Applying Lemma~\ref{lem:invert-delta} with
$$
M=\left\lfloor \frac{q(n-k+2)}{q+1}\right\rfloor
$$
gives
$$
\tau \le
\left\lfloor
\frac{L\left\lfloor \frac{q(n-k+2)}{q+1}\right\rfloor -1}{L+1}
\right\rfloor.
$$
\end{proof}

We next connect this result with the literature on higher-order MDS codes.
For this purpose, we pass to average-radius list decoding~\cite{GuruswamiNarayanan2014}.
Let $\tau\ge 0$ be a real number, and let $L\ge 1$ be an integer.
We say that $C$ is \emph{strongly $(\tau,L)$-list-decodable} if there do not exist
$$
y\in \F_q^n,
\qquad
c^{(0)},c^{(1)},\dots,c^{(L)}\in C
\quad\text{distinct},
$$
such that
$$
\sum_{j=0}^{L}\wt\bigl(y-c^{(j)}\bigr)\le (L+1)\tau.
$$
This is the usual average-radius condition.
Following Roth~\cite{Roth2022}, we say that a linear $[n,k]$ code is \emph{LD-MDS($L$)} if it is strongly
$$
\left(\frac{L(n-k)}{L+1},L\right)\text{-list-decodable}.
$$
In other words, LD-MDS($L$) means exact attainment of the average-radius generalized Singleton bound.

The same support-splitting idea gives a sharper bound for average-radius list decoding.
\begin{theorem}\label{thm:strong-ghw-obstruction}
Let $L\ge 1$, let $1\le r\le k$, and assume that
$$
L\le q^r-1.
$$
If $C$ is strongly $(\tau,L)$-list-decodable, then
$$
d_r(C)\ge \left\lfloor \frac{(L+1)\tau}{L}\right\rfloor +1.
$$
\end{theorem}

\begin{proof}
Assume, for contradiction, that
\[
d_r(C)\le \left\lfloor \frac{(L+1)\tau}{L}\right\rfloor.
\]
Choose an $r$-dimensional subcode $D\le C$ with
\[
|\supp(D)|=d_r(C).
\]
Write $S=\supp(D)$ and $s=|S|$. Then
\[
Ls\le (L+1)\tau.
\]
Since $\dim(D)=r$ and $L\le q^r-1$, the subcode $D$ contains at least $L+1$ distinct codewords. Choose
\[
c^{(0)},c^{(1)},\dots,c^{(L)}\in D.
\]
Partition $S$ as a disjoint union
\[
S=S_0\sqcup S_1\sqcup\cdots\sqcup S_L,
\]
where some parts may be empty.
Define $y\in\F_q^n$ by
\[
y_i=
\begin{cases}
c^{(j)}_i, & i\in S_j,\\
0, & i\notin S.
\end{cases}
\]
For each $j$, the codeword $c^{(j)}$ is zero outside $S$.
Also, $y$ agrees with $c^{(j)}$ on $S_j$.
Hence
\[
\wt\bigl(y-c^{(j)}\bigr)\le s-|S_j|.
\]
Summing over $j=0,1,\dots,L$ gives
\[
\sum_{j=0}^{L}\wt\bigl(y-c^{(j)}\bigr)
\le
\sum_{j=0}^{L}(s-|S_j|)
=(L+1)s-|S|=Ls\le (L+1)\tau.
\]
Thus $y$ and the $L+1$ distinct codewords $c^{(0)},c^{(1)},\dots,c^{(L)}$ violate strong $(\tau,L)$-list-decodability.
This is a contradiction.
Therefore
\[
d_r(C)\ge \left\lfloor \frac{(L+1)\tau}{L}\right\rfloor +1.
\]
\end{proof}

\begin{corollary}\label{cor:ldmds-ghw}
Let $C$ be a $q$-ary linear $[n,k]$ code.
Assume that $C$ is LD-MDS($L$).
If $1\le r\le k$ and $L\le q^r-1$, then
$$
d_r(C)\ge n-k+1.
$$
\end{corollary}

\begin{proof}
Since $C$ is LD-MDS($L$), it is strongly
$$
\left(\frac{L(n-k)}{L+1},L\right)\text{-list-decodable}.
$$
Applying Theorem~\ref{thm:strong-ghw-obstruction} with
$$
\tau=\frac{L(n-k)}{L+1}
$$
gives
$$
d_r(C)
\ge
\left\lfloor \frac{L+1}{L}\cdot \frac{L(n-k)}{L+1}\right\rfloor +1
=
\lfloor n-k\rfloor +1
=
 n-k+1.
$$
\end{proof}

\begin{corollary}\label{cor:ldmds-mds-griesmer}
Assume $1\le L\le q-1$ and $k\ge 2$.
Let $C$ be a $q$-ary linear $[n,k]$ code that is LD-MDS($L$).
Then $C$ is an MDS code.
Consequently,
$$
n \ge \sum_{i=0}^{k-1}\left\lceil \frac{n-k+1}{q^i}\right\rceil.
$$
In particular, if $k\ge 2$, then
$$
n-k+1\le q.
$$
\end{corollary}

\begin{proof}
Since $L\le q-1$, we may apply Corollary~\ref{cor:ldmds-ghw} with $r=1$.
This yields
$$
d(C)=d_1(C)\ge n-k+1.
$$
By the ordinary Singleton bound,
$$
d(C)\le n-k+1.
$$
Hence
$$
d(C)=n-k+1,
$$
so $C$ is an MDS code.

Applying the Griesmer bound (Theorem~\ref{thm:griesmer}) to $C$, we obtain
$$
n \ge \sum_{i=0}^{k-1}\left\lceil \frac{d(C)}{q^i}\right\rceil
=
\sum_{i=0}^{k-1}\left\lceil \frac{n-k+1}{q^i}\right\rceil.
$$

Now assume $k\ge 2$.
For every $i\ge 2$, the term $\left\lceil (n-k+1)/q^i\right\rceil$ is at least $1$.
Therefore
$$
n
\ge
(n-k+1)+\left\lceil \frac{n-k+1}{q}\right\rceil +(k-2).
$$
Rearranging gives
$$
\left\lceil \frac{n-k+1}{q}\right\rceil\le 1,
$$
which is equivalent to $n-k+1\le q$.
\end{proof}
\begin{remark}
Corollary~\ref{cor:ldmds-mds-griesmer} also gives an upper bound on the redundancy.
In this parameter range, it slightly improves Theorem~10 of Roth~\cite{Roth2022}.
\end{remark}

\section{A construction of optimal list-decodable linear codes}
In this section, we construct linear codes that are optimal for the Griesmer-type list-decoding bound.
These codes do not attain the Singleton-type bound.
At the relevant radius, they are genuinely list-decodable rather than uniquely decodable.

Let $q=3^a$ with $a\ge1$, let $L=2$, and set
\[
\tau=\frac{2q}{3},
\qquad
\Delta=\delta_2(\tau)=\tau+\left\lfloor\frac{\tau}{2}\right\rfloor+1=q+1.
\]
Let $\PP^1(\F_q)$ be the projective line.
Choose two distinct projective points $P_1,P_2\in\PP^1(\F_q)$.
We form a $2\times(q+3)$ generator matrix $G$ as follows.
First, take one nonzero representative column from every point of $\PP^1(\F_q)$.
Then add one extra copy of the representative for each of $P_1$ and $P_2$.
Let
\[
C=\{uG:u\in\F_q^2\}\le\F_q^{q+3}.
\]

\begin{proposition}
The code $C$ is a $q$-ary linear code with parameters
\[
[n,k,d]=[q+3,\,2,\,q+1].
\]
\end{proposition}

\begin{proof}
Because $q=3^a$, the number $\tau=2q/3$ is an integer and
\[
\left\lfloor\frac{\tau}{2}\right\rfloor
=\left\lfloor\frac{q}{3}\right\rfloor
=\frac{q}{3},
\]
so $\Delta=q+1$.
The columns of $G$ contain representatives of all points of $\PP^1(\F_q)$.
Hence they contain two linearly independent columns.
Therefore $k=2$.

Let $m(P)$ denote the multiplicity of the projective point $P$ among the columns of $G$.
Then $m(P)=2$ for $P=P_1,P_2$, and $m(P)=1$ otherwise.
For a nonzero row vector $u\in\F_q^2$, the codeword $uG$ vanishes exactly on the columns corresponding to one projective point.
We denote this point by $u^\perp\in\PP^1(\F_q)$.
Therefore
\[
\wt(uG)=n-m(u^\perp).
\]
The largest multiplicity is $2$, and it is attained at $P_1$ and $P_2$.
Thus the minimum distance is
\[
d=n-2=(q+3)-2=q+1.
\]
\end{proof}

\begin{theorem}
The code $C$ is list-decodable with radius $2q/3$ and list size $2$.
It is not list-decodable with radius $2q/3$ and list size $1$.
It attains equality in the bound of Theorem~\ref{thm:main}, and it satisfies
\[
\delta_2(2q/3)<n-k+1.
\]
\end{theorem}

\begin{proof}
The Griesmer-type lower bound in dimension $2$ is
\[
\Delta+\ceil{\frac{\Delta}{q}}
=(q+1)+\ceil{1+\frac1q}=q+3=n.
\]
Thus equality holds in the bound of Theorem~\ref{thm:main}.
The Singleton-type bound is not tight for this family because
\[
n-k+1=(q+3)-2+1=q+2>q+1=\Delta.
\]

We now prove $(2q/3,2)$-list-decodability.
Put $\rho=2q/3$.
Suppose, for contradiction, that some Hamming ball of radius $\rho$ contains three distinct codewords.
After translating by one of these codewords, we obtain a center $z\in\F_q^n$ and two distinct nonzero codewords $A=uG$ and $B=vG$ such that
\[
0,A,B\in B(z,\rho).
\]
In particular,
\[
\wt(z)+\wt(z-A)+\wt(z-B)\le 3\rho=2q.
\]
We show that the left-hand side is always at least $2q+1$.

First assume that $u$ and $v$ are linearly dependent.
Then $B=\lambda A$ for some $\lambda\in\F_q\setminus\{0,1\}$.
On every coordinate in $\supp(A)$, the three symbols $0,A_i,B_i$ are distinct.
Hence, at such a coordinate, any center symbol disagrees with at least two of these three symbols.
Since $\wt(A)\ge d=q+1$, we get
\[
\wt(z)+\wt(z-A)+\wt(z-B)
\ge 2\wt(A)
\ge 2q+2,
\]
contradicting the preceding upper bound.

Now assume that $u$ and $v$ are linearly independent.
For each $P\in\PP^1(\F_q)$, fix the representative column $g_P$ used in $G$.
The linear map
\[
\F_q^2\longrightarrow \F_q^2,
\qquad x\longmapsto (u\cdot x,\,v\cdot x),
\]
is invertible.
Therefore, as $P$ ranges over $\PP^1(\F_q)$, the pair of symbols $(u\cdot g_P,v\cdot g_P)$ ranges projectively over all of $\PP^1(\F_q)$.
The triple of symbols
\[
0,\quad u\cdot g_P,\quad v\cdot g_P
\]
has fewer than three distinct values exactly when
\[
u\cdot g_P=0,
\qquad
v\cdot g_P=0,
\qquad\text{or}\qquad
(u-v)\cdot g_P=0.
\]
These are three distinct projective points.
Indeed, $u$, $v$, and $u-v$ are nonzero, and no two of them are scalar multiples.

Let $M$ be the total column multiplicity coming from these three special projective points.
Let $E=n-M$ be the total column multiplicity coming from the nonspecial projective points.
Each special point contributes its usual one column.
Among all projective points, only $P_1$ and $P_2$ receive one extra copy.
Therefore
\[
M\le 3+2=5,
\qquad
E=n-M\ge(q+3)-5=q-2.
\]
At a special coordinate, the three symbols $0,A_i,B_i$ have exactly two distinct values.
Thus every center symbol disagrees with at least one of them.
At a nonspecial coordinate, the three symbols are distinct.
Thus every center symbol disagrees with at least two of them.
Consequently
\[
\wt(z)+\wt(z-A)+\wt(z-B)
\ge M+2E=n+E\ge(q+3)+(q-2)=2q+1.
\]
This again contradicts the upper bound $2q$.
Therefore no Hamming ball of radius $2q/3$ contains three codewords.
Hence $C$ is $(2q/3,2)$-list-decodable.

Finally, we prove that $C$ is not $(2q/3,1)$-list-decodable.
Choose a minimum-weight nonzero codeword $A\in C$.
Then $\wt(A)=q+1$.
Since $q$ is odd, $(q+1)/2$ is an integer, and
\[
\frac{q+1}{2}\le \frac{2q}{3}
\]
for every $q\ge3$.
Split the support of $A$ into two subsets of size $(q+1)/2$.
Let a center $z$ agree with $A$ on one subset and with $0$ on the other subset.
Set $z_i=0$ outside $\supp(A)$.
Then both $0$ and $A$ lie in the Hamming ball of radius $2q/3$ centered at $z$.
Thus $C$ is not $(2q/3,1)$-list-decodable.
\end{proof}

\section{Conclusion}
We developed a Griesmer-type bound for list-decodable linear codes in the small-list-size regime.
The key step is a generalized Hamming weight obstruction.
This obstruction shows that list decodability forces suitable generalized Hamming weights to be large.
In particular, for a $q$-ary linear code of minimum distance $d$ and list size $1\le L\le q-1$, we obtain
\[
d\ge \tau+\left\lfloor\frac{\tau}{L}\right\rfloor+1.
\]
This condition plays the same role for list decoding that $d\ge 2\tau+1$ plays for unique decoding.
Combining it with the classical Griesmer bound gives a $q$-dependent lower bound on the block length.

We also gave an explicit family of $q$-ary linear $[q+3,2,q+1]$ codes for $q=3^a$.
These codes are $(2q/3,2)$-list-decodable.
They attain the Griesmer-type bound with equality, but they fail to attain the Singleton-type bound.
Thus the Griesmer-type bound is tight for this family.
It is also a genuine strengthening of the Singleton-type obstruction.

\end{document}